\documentclass[
final
]{dmtcs-episciences}

\usepackage[utf8]{inputenc}
\usepackage{subfigure}

\newtheorem{theorem}{Theorem}[section]

\newtheorem{lemma}[theorem]{Lemma}
\newtheorem{corollary}[theorem]{Corollary}
\newtheorem{proposition}[theorem]{Proposition}

\newtheorem{openProblem}{Question}

\usepackage[round]{natbib}

\author{T\i naz Ekim\affiliationmark{1}
  \and Didem Gözüpek\affiliationmark{4}
  \and Ademir Hujdurovi\'c \affiliationmark{2,3}
  \and Martin Milani\v{c} \affiliationmark{2,3} \thanks{This work is supported in part by the Slovenian Research Agency (I0-0035, research program P1-0285 and research projects N1-0032, N1-0038, N1-0062, J1-7051, and J1-9110). T.~Ekim is supported by Turkish Academy of Sciences GEBIP award.} }
\title[On Almost Well-Covered Graphs of Girth at Least $6$]{On Almost Well-Covered Graphs of Girth at Least $6$}
\affiliation{
  Department of Industrial Engineering, Bo\u{g}azi\c{c}i University, Istanbul, Turkey\\
  University of Primorska, UP IAM, Muzejski trg 2, SI-6000 Koper, Slovenia\\
  University of Primorska, UP FAMNIT, Glagolja\v ska 8, SI-6000 Koper, Slovenia\\
  Department of Computer Engineering, Gebze Technical University, Kocaeli, Turkey}
\received{2018-5-18}
\revised{2018-10-26}
\accepted{2018-11-2}
\begin{document}
\publicationdetails{20}{2018}{2}{17}{4514}
\maketitle
\begin{abstract}
  We consider a relaxation of the concept of well-covered graphs, which are graphs with all maximal independent sets of the same size.
The extent to which a graph fails to be well-covered can be measured by     its \emph{independence gap}, defined as the difference between the maximum and minimum sizes of a maximal independent set in $G$. While the well-covered graphs are exactly the graphs of independence gap zero, we investigate in this paper graphs of independence gap one, which we also call \emph{almost well-covered} graphs. Previous works due to Finbow et al.~and~Barbosa et al.~have implications for the structure of almost well-covered graphs of girth at least $k$ for $k\in \{7,8\}$.
We focus on almost well-covered graphs of girth at least~$6$. We show that every graph in this class has at most two vertices each of which is adjacent to exactly $2$ leaves. We give efficiently testable characterizations of almost well-covered graphs of girth at least~$6$ having exactly one or exactly two such vertices. Building on these results, we develop a polynomial-time recognition algorithm of almost well-covered $\{C_3,C_4,C_5,C_7\}$-free graphs.
\end{abstract}

\section{Introduction}

A graph is said to be \emph{well-covered} if all its (inclusion-)maximal independent sets have the same size. Well-covered graphs were introduced by Plummer in 1970~\cite{MR0289347} and have been studied extensively in the literature, see, e.g., the survey papers~\cite{MR1254158,MR1677797}. One of the main motivations for the study of well-covered graphs stems from the fact that the maximum independent set problem, which is generally {\sf NP}-complete, can be solved in polynomial time in the class of well-covered graphs by a greedy algorithm.
This point of view motivated Caro et al.~to study a more general concept, the so-called `greedy hereditary systems'~\cite{MR1368720}.

Here, we consider a different generalization. For a graph $G$ we denote by $\alpha(G)$ the maximum size of an independent set in $G$ and by $i(G)$ the minimum size of a maximal independent set in $G$. We say that the \emph{independence gap} of $G$ is the difference $\alpha(G)-i(G)$ and we denote it by $\mu_\alpha(G)$. An analogous parameter for maximal matchings was introduced and studied recently by Deniz et al.~\cite{MR3720285} under the name \emph{matching gap} of a graph, denoted by $\mu(G)$. Since matchings in a graph $G$ are exactly the independent sets in its line graph, the study of the independence gap can be seen as a generalization of the study of the matching gap.

Clearly, a graph $G$ is well-covered if and only if its independence gap is zero. In view of the extensive literature on the class of well-covered graphs, it is natural to try to generalize results on well-covered graphs to graphs of bounded independence gap. Again, identification of such graph classes is motivated by the maximum independent set problem:
In every class of graphs of uniformly bounded independence gap, the maximum independent set problem can be efficiently approximated up to an additive constant by the greedy algorithm.

We focus in this paper on the first case beyond the well-covered graphs. We say that a graph $G$ is \emph{almost well-covered} if it is of unit independence gap, that is, if $\mu_\alpha(G)= 1$.
It is in general difficult to determine if a given graph is almost well-covered. Indeed, the fact that recognizing well-covered graphs is {\sf co-NP}-complete (shown independently by Chv\'atal and Slater~\cite{MR1217991} and by Sankaranarayana and Stewart~\cite{MR1161178}) implies that recognizing almost well-covered graphs is {\sf NP}-hard. It suffices to observe that a graph $G$ is well-covered if and only if the disjoint union of $G$ with the three-vertex path is almost well-covered.

Given the intractability of the problem in general, it is of interest to characterize almost well-covered graphs under some additional restrictions.
In 1994, Finbow et al.~\cite{awc_girth8} denoted by $M_k$ the class of graphs that have maximal independent sets of exactly $k$ different sizes and characterized graphs in $M_2$ of girth at least $8$. Clearly, every almost well-covered graph is in $M_2$, and, in fact, one can use the result of Finbow et al.~to derive a characterization of almost well-covered graphs of girth at least 8, which also implies that this class of graphs can be recognized in polynomial time.
In 1998, Barbosa and Hartnell~\cite{MR1676478} characterized almost well-covered simplicial graphs and gave a sufficient condition for a chordal graph to be almost well-covered. (They denoted the class of almost well-covered graphs by $I_2$.) More recently, graphs in $M_k$ were studied further by Hartnell and Rall~\cite{MR3053598} and by
 Barbosa et al.~\cite{MR3056978}.
  A result due to Barbosa et al.~\cite[Theorem 2]{MR3056978} implies that for every $d$, there are only finitely many connected almost well-covered graphs of minimum degree at least $2$, maximum degree at most $d$, and girth at least~$7$.
  
  Let us also remark that, contrary to the fact that graphs of zero matching gap (known as equimatchable graphs) can be recognized in polynomial time~\cite{MR3128394,MR777180}, the computational complexity of recognizing graphs of unit matching gap is open. In terms of almost well-covered graphs, this means that the complexity of determining if a given line graph is almost well-covered is open.

\bigskip
\noindent{\bf Our results.}
The main goal of this paper is to further the study of almost well-covered graphs, with a focus on girth conditions. As noted above, results of Finbow et al.~\cite{awc_girth8} and of Barbosa et al.~\cite{MR3056978} can be used to infer results about almost well-covered graphs of girth at least $8$ or at least~$7$, respectively. We study almost well-covered graphs of girth at least~$6$. We obtain three results, which can be summarized as follows.

A vertex is said to be of \emph{type $2$} if it is adjacent to exactly $2$ leaves. We first show that every almost well-covered graph of girth at least~$6$ has at most two vertices of type $2$.
As our first main result, we give a complete structural characterization of
almost well-covered graphs of girth at least~$6$ having exactly two vertices of type 2 (Theorem~\ref{thm:2G2Girth6}).
Next, we characterize almost well-covered graphs of girth at least~$6$ having exactly one vertex of type 2 (Theorem~\ref{thm:1type2}).
For both cases, we develop a polynomial-time recognition algorithm for graphs in the respective class (Corollary~\ref{cor:polyrecognition2type2} and Theorem~\ref{thm:1type2Poly}). Finally, we use these results to develop a polynomial-time recognition algorithm of almost well-covered $C_7$-free graphs of girth at least $6$, or, equivalently, of almost well-covered $\{C_3,C_4,C_5,C_7\}$-free graphs (Theorem~\ref{thm:polyrecognition}).

\medskip
\noindent{\bf Structure of the paper.} In Section~\ref{sec:general}, we collect the main definitions and develop several technical lemmas about the independence gap of a graph for later use. Section~\ref{sec:no-short-cycles} is split into three subsections, in which our main results are derived.


\section{General results on independence gap}\label{sec:general}

For some of our theorems and proofs, it will be convenient to allow working with the null graph, the graph without vertices; clearly, if $G$ is the null graph, then $i(G) = \alpha(G) = 0$. Given two vertex sets $A$ and $B$ in a graph $G$, we say that $A$ \emph{dominates} $B$ if every vertex in $B$ has a neighbor in $A$. A \emph{clique} in a graph is a set of pairwise adjacent vertices, and a clique is said to be \emph{maximal} if it is not contained in any larger clique. As usual, we denote by $P_n$, $C_n$, and $K_n$ the path, the cycle, and the complete graph of order $n$, respectively.
Given a set of graphs ${\cal F}$ and a graph $G$, we say that $G$ is ${\cal F}$-free if no induced subgraph of $G$ is isomorphic to a member of ${\cal F}$; if ${\cal F} = \{F\}$, we also say that $G$ is $F$-free.
For graph theoretic terms not defined here, we refer to~\cite{MR1367739}.

The following straightforward lemma reduces the problem of determining the independence gap of a graph to its (connected) components.

\begin{lemma}\label{connected}
Let $G$ be a graph with components $H_1,\ldots, H_k$.
Then $\mu_\alpha(G) = \sum_{j = 1}^k\mu_\alpha(H_j)$.
\end{lemma}
\begin{proof}
Immediate from the facts that
$\alpha(G) = \sum_{j = 1}^k\alpha(H_j)$ and
$i(G) = \sum_{j = 1}^ki(H_j)$.
\end{proof}

\begin{sloppypar}
Lemma~\ref{connected} shows, in particular, that in any class of graphs closed under taking components and in which testing whether the graph is well-covered can be done efficiently, the problem of characterizing almost well-covered graphs reduces to the problem of characterizing almost well-covered connected graphs in the class.
\end{sloppypar}

\begin{corollary}\label{cor:components}
A graph is almost well-covered if and only if all its components are well-covered, except one, which is almost well-covered.
\end{corollary}

Since every complete graph $G$ has $\mu_\alpha(G)= 0$, Lemma~\ref{connected} also has the following consequence.

\begin{corollary}\label{cor:clique-components}
Let $G$ be a graph and let $U$ be the set of vertices $u\in V(G)$ such that the component of $G$ containing $u$ is complete.
Then, $\mu_\alpha(G) = \mu_\alpha(G-U)$.
\end{corollary}

The following fact is often used in our proofs. Although it is implied by~Lemma 2.1~in~\cite{awc_girth8}, we give a short proof for the sake of completeness.

\begin{lemma}\label{tool}
For every independent set $I$ in $G$ we have $\mu_\alpha(G-N[I])\leq \mu_\alpha(G)$.
\end{lemma}

\begin{proof}
Let $k = \mu_\alpha(G)$ and assume for a contradiction that $G-N[I]$ has independence gap at least $k+1$. Then it contains two maximal independent sets $I_1$ and $I_2$ of sizes differing by at least $k+1$. Clearly, the two independent sets obtained by adding $I$ to each one of $I_1$ and $I_2$ are maximal in $G$ and have sizes differing by at least $k+1$, contradicting with  $\mu_\alpha(G)= k$.
\end{proof}

We now assign a non-negative integer, called the \emph{type} of $v$, to some vertices $v$ of $G$, see Fig.~\ref{fig:types-example} for an illustration.
Let $U$ denote the set of all vertices $u\in V(G)$ such that the component of $G$ containing $u$ is complete. Vertices of $G-U$ are classified according to their type as follows.
The vertex set of $G-U$ is split into the leaves (vertices of degree~$1$) and non-leaves, which we call \emph{internal vertices}.
Internal vertices adjacent to exactly $k$ leaves will be of \textit{type $k$}. Note that a leaf in a component other than $K_2$
is not assigned any type and there are two kinds of type $0$ vertices: vertices contained in a complete component and internal vertices
adjacent to no leaf. Consequently, any vertex in $U$ is of type $0$. We further denote by $G_i$ the subgraph of $G$ induced by all vertices of type $i$.
For an internal vertex $v\in V(G)$, a \emph{leaf of $v$} is a leaf adjacent to $v$.

\medskip

\begin{figure}[htpb]
\begin{center}
\includegraphics[width=0.6\linewidth]{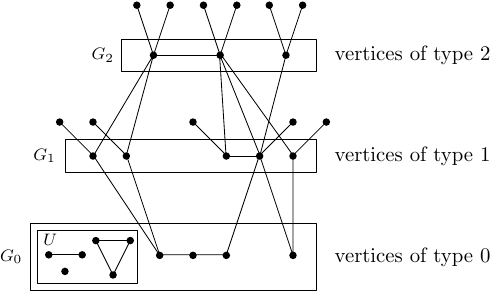}
\caption{Illustration of the types of vertices in a graph.}\label{fig:types-example}
\end{center}
\end{figure}

The next lemma shows that the number of leaves a vertex can have is bounded above by the independence gap plus one.

\begin{lemma}\label{lem:leafcondition}
Let $G$ be a graph with $\mu_\alpha(G)\leq k$. Then every internal vertex of $G$ is adjacent to at most $k+1$ leaves.
\end{lemma}

\begin{proof}
Assume for a contradiction that $x$ is an internal vertex with $p\ge k+2$ leaves $b_1,\ldots, b_{p}$. Then a maximal independent set $I$ containing $x$ and a maximal independent set $I'$ obtained by extending independent set $(I-\{x\}) \cup \{b_1, \ldots ,b_{p} \}$ to a maximal one have their sizes differing by at least $k+1$, a contradiction with $\mu_\alpha(G)\leq k$.
\end{proof}

For later use, we record the following immediate consequence of Lemma~\ref{lem:leafcondition}.

\begin{corollary}\label{cor:atmost2leaves}
In a graph $G$ with $\mu_\alpha(G)\leq 1$, every internal vertex is adjacent to at most two leaves.
\end{corollary}

The following lemma is similar to~\cite[Lemma 2.10]{awc_girth8}.

\begin{lemma} \label{lem:type 2 clique}
Let $G$ be an almost well-covered graph. Then $G_2$ is a complete graph.
\end{lemma}

\begin{proof}
Let $G$ be an almost well-covered graph and $x$ and $y$ be two vertices of Type 2 which are non-adjacent. Let $I$ be a maximal independent set of $G$ containing both $x$ and $y$. Then $(I-\{x,y\} \cup \{b_1,b_2,b_3,b_4\})$  where $b_i, i=1, \ldots ,4$ are the leaves of $x$ and $y$, is an  independent set of $G$ with two more vertices than $I$, contradicting with $G$ having independence gap~$1$.
\end{proof}

The following technical claim will also be used later.

\begin{lemma}\label{lem:Hdomination}
Let $G$ be a graph, $H$ a non-null subgraph of $G$, and let $n_0$ denote the maximum number of vertices in a component of $H$. For every independent set $I\subseteq V(G)\setminus V(H)$ and every $k\leq \alpha(H)-i(H-N(I))$  there exists an independent set $I_1\subseteq I$ of size $|I_1|\leq kn_0$ such that $k\leq \alpha(H)-i(H-N(I_1))$.
\end{lemma}

\begin{proof}
Let $I$ be an independent set such that $I\subseteq V(G)\setminus V(H)$ and let $k$ be a non-negative integer such that $k\leq \alpha(H)-i(H-N(I))$.
Denote the components of $H$ by $H_1,\ldots,H_m$. It is clear that every maximal independent set in $H$ is a union of maximal independent sets in $H_j$; therefore, $\alpha(H)=\sum_{j=1}^m\alpha(H_j)$ and $i(H-N(I))=\sum_{j=1}^m i(H_j-N(I))$. Clearly, there exists a set $J\subseteq \{1,\ldots,m\}$ with $|J|\leq k$ such that $\sum_{j\in J}\left(\alpha(H_j)-i(H_j-N(I))\right)\geq k$. We may assume without loss of generality that $J= \{1,\ldots,k\}$.

For $j\in \{1,\ldots,k\}$ let $A_j=N(I)\cap V(H_j)$.
Let $A=\cup_{j=1}^{k}A_j$. It is clear that $|A|\leq kn_0$. Let $\mathcal{A}=\{N(x)\cap A\mid x \in I\}$.
Then $\mathcal{A}$ is a covering family of $A$, that is, $A$ is the union of the sets in $\mathcal{A}$.
Let $\mathcal{A}_1\subseteq \mathcal{A}$ be any inclusion-minimal subfamily covering $A$.
By the minimality of $\mathcal{A}_1$, every set in $\mathcal{A}_1$ covers at least one vertex of $A$ that is not covered by other sets of $\mathcal{A}_1$ and therefore $|\mathcal{A}_1|\leq |A|$. Furthermore, by definition of $\mathcal{A}_1$, each set $S\in \mathcal{A}_1$ can be identified by an element $x\in I$ such that $N(x)\cap A=S$. This defines a subset $I_1$ of $I$ such that $|I_1|\leq |\mathcal{A}_1|$ and $N(I_1)\cap V(H_j)=A_j$, for every $j\in \{1,\ldots,k\}$. It is now clear that $H_j-N(I)=H_j-N(I_1)$, for every $j\in \{1,\ldots,k\}$.
Hence
$$\alpha(H)-i(H-N(I_1))\geq \sum_{j=1}^{k}\left(\alpha(H_j)-i(H_j-N(I_1))\right)=\sum_{j=1}^{k}\left(\alpha(H_j)-i(H_j-N(I))\right)\geq k.$$
Since $|I_1|\le |\mathcal{A}_1|\le |A|\le kn_0$, this concludes the proof.
\end{proof}

Next, we develop an important tool characterizing graphs with independence gap at most $k$ among graphs with internal vertices of types 0 and 1 only. Note that if $G$ is a graph with internal vertices of types 0 and 1 only and $G_0$ is the null graph, then $G$ is well-covered.

\begin{theorem}\label{thm:gapk_characterization}
Let $G$ be a graph with internal vertices of types 0 and 1 only such that $G_0$ is non-null, and let $n_0$ denote the maximum size of a component of $G_0$. Then $\mu_\alpha(G)\leq k$ if and only if $\mu_\alpha(G_0)\leq k$ and for every independent set $I$ in $G_1$ of size at most $(k+1)n_0$, it holds that
$$
\alpha(G_0)-i(G_0-N(I))\leq k.
$$
\end{theorem}

\begin{proof}
Let $U$ be the set of all vertices $u\in V(G)$ such that the component of $G$ containing $u$ is complete.

\medskip
Suppose that $\mu_\alpha(G)\leq k$.
Let $G' = G-U$. By Corollary~\ref{cor:clique-components}, $\mu_\alpha(G') \le k$.
By Lemma \ref{tool} we also have  $\mu_\alpha(G_0)\leq k$ since $G_0=G-N[L]$, where $L$ is the set of all leaves of $G-U$, thus an independent set. Now assume for a contradiction that there is an independent set $I$ of $G_1$ of size at most $(k+1)n_0$
such that $\alpha(G_0)-i(G_0-N(I))\geq k+1$. Now let $I_0$ be a maximal independent set of $G_0-N(I)$ of size $i(G_0-N(I))$
and $I_0^*$ be a maximum independent set of $G_0$.
Consider now the two independent sets $I_1=I_0^*\cup L$ and $I_2=I_0\cup I \cup (L\setminus N(I))$.
Maximality of both $I_1$ and $I_2$ in $G$ follow from their definitions.
We have $|I_1|=\alpha(G_0)+|L|$ and $|I_2|=i(G_0-N(I))+|L|$, yielding $|I_1|-|I_2|=\alpha(G_0)-i(G_0-N(I))\geq k+1$, contradiction with  $\mu_\alpha(G)\leq k$.

\medskip
Suppose now that $\mu_\alpha(G_0)\leq k$ and for every independent set $I$ in $G_1$ of size at most $n_0(k+1)$ we have $\alpha(G_0)-i(G_0-N(I))\leq k$.
Suppose for a contradiction that $\mu_\alpha(G)\geq k+1$, that is, there are maximal independent sets $I_1$ and $I_2$ of $G$ such that $|I_1|-|I_2|\geq k+1$. We define $I_1^0=I_1\cap V(G_0)$, $I_2^0=I_2\cap V(G_0)$ and $I_2^1=I_2\cap V(G_1)$. Since $G$ has only vertices of type 0 and 1, we have that $V(G)$ is the disjoint union of $V(G_0)$, $V(G_1)$, and $L$, where $L$ is the set of leaves of $G-U$. It is clear that both $I_1$ and $I_2$ have exactly $|V(G_1)|$ vertices in common with $V(G_1) \cup L$. This implies that $|I_1^0|-|I_2^0|\geq k+1$. Now we have $\alpha(G_0)\geq |I_1^0|\geq |I_2^0|+k+1 \geq i(G_0-N(I_2^1))+k+1$, where the last inequality holds because $I_2^0$ is a maximal independent set of $G_0-N(I_2^1)$. By Lemma~\ref{lem:Hdomination} with $G-L$, $k+1$, $G_0$, and $I_2^1$ in place of $G$, $k$, $H$, and $I$, respectively, it follows that
there exists an independent set $I'\subseteq I_2^1$ in $G_1$ of size at most $n_0(k+1)$ such that $\alpha(G_0) - i(G_0-N(I'))\ge k+1$.
This is a contradiction to our assumption that $\alpha(G_0)-i(G_0-N(I))\leq k$ for every independent set $I$ in $G_1$.
\end{proof}

The implication of Theorem \ref{thm:gapk_characterization} to graphs with independence gap $0$ or at most $1$ is worth mentioning separately.
\begin{corollary}\label{cor:type1main}
Let $G$ be a graph with internal vertices of types 0 and 1 only, and
let $n_0$ denote the maximum size of a component of $G_0$
(with $n_0 = 0$ if $G_0$ is null). Then
\begin{enumerate}
\item $G$ is well-covered if and only if $G_0$ is well-covered and  for every independent set $I$ in $G_1$ of size at most $n_0$, we have $i(G_0-N(I))= \alpha(G_0)$.
\item $\mu_\alpha(G)\leq 1$  if and only if $\mu_\alpha(G_0)\leq 1$ and  for every independent set $I$ in $G_1$ of size at most $2n_0$, we have $\alpha(G_0)-i(G_0-N(I))\leq 1$.
\end{enumerate}
\end{corollary}

Another consequence of Theorem~\ref{thm:gapk_characterization} is the following.

\begin{corollary}\label{cor:G_0boundedPoly}
Let $G$ be a graph with internal vertices of types 0 and 1. If the size of every component of $G_0$ is bounded by a constant, then
for every constant $k\ge 0$, testing whether $\mu_\alpha(G)\leq k$ can be done in polynomial time.
\end{corollary}

\begin{proof}
We assume that $G_0$ is non-null (since otherwise, as observed above, $G$ is well-covered and hence $\mu_\alpha(G)\le k$ holds for all non-negative $k$).
Suppose that every component of $G_0$ has at most $c$ vertices and let $H_1,\ldots,H_m$ be the components of $G_0$. Note that $G_0$, $G_1$, and the components of $G_0$ can be computed in linear time.
By Theorem~\ref{thm:gapk_characterization} it follows that $\mu_\alpha(G)\leq k$ if and only if $\mu_\alpha(G_0)\leq k$ and for every independent set $I$ in $G_1$ of size at most $(k+1)c$, it holds that $\alpha(G_0)-i(G_0-N(I))\leq k$.

Since the number of vertices of $H_j$ is bounded by a constant, the values of $\alpha(H_j)$ and $i(H_j)$ can be computed in constant time. Furthermore, since $\alpha(G_0)=\sum_{j=1}^m\alpha(H_j)$ and $i(G_0)=\sum_{j=1}^m i(H_j)$, it follows that verifying if $\mu_\alpha(G_0)=\alpha(G_0)-i(G_0)\leq k$ can be done in linear time.

Similarly, since there are at most $\mathcal{O}(|V(G_1)|^{(k+1)c})$ independent sets of $G_1$ of size at most $(k+1)c$, and for each such set $I$ the values $i(H_j-N(I))$ can be computed in constant time, the result follows.
\end{proof}

\section{Almost well-covered graphs without short cycles}\label{sec:no-short-cycles}

In this section we develop our main results, which relate to almost well-covered graphs without short cycles. In our proofs we will make use of the following result, characterizing well-covered graphs of girth at least $6$. Recall that a \emph{perfect matching} in a graph is a set of pairwise disjoint edges such that every vertex
of the graph is an endpoint of one of them, and that an edge is said to be \emph{pendant} if it is incident with a vertex of degree~$1$.

\begin{proposition}[\cite{MR1198396}]
\label{prop:wc_girth6}
Let $G$ be a connected graph of girth at least $6$ isomorphic to neither $C_7$ nor $K_1$.
Then $G$ is well-covered if and only if its pendant edges form a perfect matching.
\end{proposition}

Recall also that due to Corollary~\ref{cor:atmost2leaves}, in the study of almost well-covered graphs we may assume that the graph under consideration has no type $k$ vertices for any $k\ge 3$. If $G$ is an almost well-covered graph of girth at least $6$, then $G$ is triangle-free and by Lemma~\ref{lem:type 2 clique}, it follows that there are either $2$, $1$, or $0$ vertices of type $2$. The first two of these three cases are characterized in the next two subsections. Building on these characterizations, we consider in Section~\ref{subsec:notype2} the general case, but with the additional assumption that the $7$-cycle is also forbidden.

\subsection{Graphs of girth at least 6 with exactly two vertices of type 2}\label{subsec:2type2}

The following lemma shows that
in an almost well-covered graph of girth at least 6
there are no edges between vertices of type 2 and type 0.

\begin{lemma}\label{lem:key2G_2}
Let $G$ be an almost well-covered graph of girth at least 6. Then no vertex in $G_2$ is adjacent to a vertex in $G_0$.
\end{lemma}

\begin{proof}
Suppose that $G$ is an almost well-covered graph and $x\in V(G_2)$ with $N(x)\cap V(G_0)\neq \emptyset$. Since $x\in V(G_2)$, it follows that $x$ is adjacent with two leaves, say $x_1$ and $x_2$. Let $y\in N(x)\cap V(G_0)$. Let $I$ be the set of vertices at distance 3 from $x$ and at distance 2 from $y$. Then $I$ is an independent set since $G$ is of girth at least 6. Consider now the graph $G'=G-N[I]$. Observe that $y$ is a leaf in $G'$ because $y$ is an internal vertex such that all of its neighbors apart from $x$ are in $N(I)$. Therefore, $x$ is adjacent to at least three leaves in $G'$, namely, $x_1$, $x_2$, and $y$. By  Lemma~\ref{lem:leafcondition}, it follows that
 $\mu_\alpha(G')\geq 2$, which by Lemma~\ref{tool} contradicts the assumption that $G$ is almost well-covered.
\end{proof}

Proposition~\ref{prop:wc_girth6} and Corollary~\ref{cor:components} imply that the study of
almost well-covered graphs of girth at least 6 with exactly two type 2 vertices reduces to the connected case.

\begin{theorem}\label{thm:2G2Girth6}
Let $G$ be a connected graph of girth at least 6, with exactly two vertices $x,y$ of type $2$, and with no type $k$ vertices for $k\geq 3$.
Then $G$ is almost well-covered if and only if $x$ and $y$ are adjacent and one of the following two conditions holds:
\begin{enumerate}
  \item $V(G_0) = \emptyset$;
  \item $G_0\cong K_2$, neither of $x$ and $y$ has a neighbor in $G_0$, and
  the two vertices of $G_0$
  are of degree~$2$ in $G$ and
  are contained in an induced $6$-cycle containing $x$ and $y$.
\end{enumerate}
\end{theorem}

\begin{proof}
$(\Rightarrow)$
Suppose first that $G$ is an almost well-covered graph satisfying the hypotheses of the theorem.
By Lemma \ref{lem:type 2 clique}, $x$ and $y$ are adjacent.
From now on, we assume that $V(G_0)\neq \emptyset$, since otherwise condition \textit{1.}~holds.
Let $x',x''$ and $y',y''$ be the two leaves adjacent to $x$ and $y$, respectively.
Let $A$ and $B$ be the (disjoint) sets of the remaining neighbours of $x$ and $y$, respectively.
Then by Lemma~\ref{lem:key2G_2}, each $v\in A\cup B$ is of type 1. Let $A'$ and $B'$ be the sets of leaves adjacent to vertices from $A$ and $B$, respectively. Let $I=\{x',x'',y',y''\}\cup A'\cup B'$. Observe that $I$ is an independent set in $G$. Let $G'=G-N[I]$.
See Fig.~\ref{fig:2type2-1}$(i)$ for an illustration.

\begin{figure}[htpb]
\begin{center}
\includegraphics[width=\linewidth]{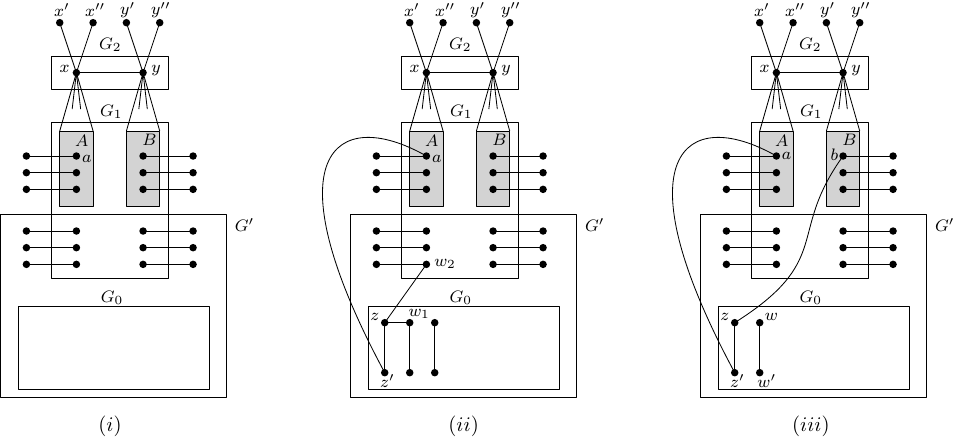}
\caption{Illustration of the proof of Theorem~\ref{thm:2G2Girth6}.}\label{fig:2type2-1}
\end{center}
\end{figure}

Observe that by definition of $G'$, for every vertex $u\in V(G')\cap V(G_1)$, its unique leaf neighbor in $G$ is a vertex of $G'$. Moreover, the vertex set of $G'$ consists exactly of the vertices in $V(G')\cap V(G_1)$, their leaves in $G$, and all the vertices in $G_0$.

We claim that $G'$ is well-covered. Suppose that $I'_1$ and $I'_2$ are two maximal independent sets in $G'$ with $|I_1'|>|I_2'|$. Let $I_1=I_1'\cup A'\cup  B'\cup \{x',x'',y',y''\}$ and $I_2=I_2'\cup A'\cup  B' \cup \{x,y',y''\}$. It is clear that $I_1$ and $I_2$ are maximal independent sets in $G$ and $|I_1|-|I_2|\geq 2$. This contradicts the assumption that $G$ is almost well-covered, and hence $G'$ is well-covered. Observe that $\alpha(G)=\alpha(G')+|A|+|B|+4$.

Since $G'$ is well-covered and the girth of $G'$ is at least 6, by Proposition~\ref{prop:wc_girth6}, every component of $G'$ is isomorphic to $K_1$, $C_7$, or to a graph with a perfect matching formed by pendant edges. Since the girth of $G$ is at least 6, it is easy to see that no component of $G'$ is isomorphic to $K_1$.
We claim that there is no component of $G'$ isomorphic to $C_7$. Suppose to the contrary that vertices $v_1,v_2,v_3,v_4,v_5,v_6,v_7$ form a 7-cycle $C$ in $G'$ (in this order).  Since $G$ is connected, we may assume without loss of generality~there exists a vertex $a\in A$ such that $a$ is adjacent to $v_1$. Due to the girth condition, $a$ is not adjacent to any other vertex in the 7-cycle $C$. Extend $\{v_3,v_6\}$ to  a maximal independent set  $J'$ in $G'-N(a)$. Observe that $|J'|<\alpha (G')$ since $J'$ contains exactly two vertices from the 7-cycle $C$, while any maximal independent set of $G'$ contains three vertices from $C$. Extend $J'\cup \{a,y\}$ to a maximum independent set $J$ in $G$. It is now easy to see that $|J|<\alpha(G')+|A|+|B|+3=\alpha(G)-1$. This contradicts the assumption that $G$ is almost well-covered and proves our claim that no component of $G'$ is isomorphic to $C_7$.

Since no component of $G'$ is isomorphic to a $K_1$ or a $C_7$, we infer that $G'$ has a perfect matching $M$ formed by pendant edges.
Clearly, every edge connecting a vertex in $V(G')\cap V(G_1)$ with its leaf neighbor is in $M$.

Next, we show that every component of $G_0$ is a component of $G'$ isomorphic to a $K_2$.
To show this, it suffices to show that every vertex of $G_0$ is a leaf in $G'$. Suppose that this is not the case. Then, there exists an edge $zz'\in M$ such that $z,z'\in V(G_0)$, vertex $z'$ is a leaf in $G'$, and there exists a vertex $w\neq z'$ in $G'$ such that $z$ is adjacent to $w$.
Note that $w$ belongs to either $V(G_0)$ or to $V(G')\cap V(G_1)$.
We may assume without loss of generality that there exists a vertex $a\in A$ such that $z'$ is adjacent to $a$ (see Fig.~\ref{fig:2type2-1}$(ii)$; the two possibilities for $w$ are denoted by $w_1$ and $w_2$). The girth condition implies that $w$ and $a$ are not adjacent in $G$.
Let $I$ be a maximal independent set in $G$ containing $\{a,w,y\}$. It is clear that $|I\cap V(G')|<\alpha(G')$ since neither $z$ nor $z'$ is in $I$; therefore, $|I|=|I\cap V(G')|+ |A|+|B|+3<\alpha(G)-1$, contradicting the assumption that $G$ is almost well-covered.

Let $z$ and $z'$ be two adjacent vertices in $G_0$. Since the component of $G_0$ containing $z$ and $z'$ is a component of $G'$ isomorphic to $K_2$ and $z$ and $z'$ are internal vertices of $G$, each of them has a neighbor in $A\cup B$. As above, we may assume without loss of generality that there exists a vertex $a\in A$ such that $z'$ is adjacent to $a$. Moreover, there exists a neighbour $b$ of $z$ within $A\cup B$. The girth condition implies that $b\in B$. Since the girth of $G$ is at least $6$, it follows that $z$ and $z'$ have no other neighbours in $A\cup B$, hence they are of degree $2$ in $G$, and they lie on an induced $6$-cycle containing $x$ and $y$. See Fig.~\ref{fig:2type2-1}$(iii)$ for an illustration.

To complete the proof of the forward direction of the implication, it remains to show that $V(G_0) =\{z,z'\}$.
Suppose that this is not the case and consider a pair $w,w'$ of adjacent vertices in $G_0-\{z,z'\}$.
Since the component of $G_0$ containing $w$ and $w'$ is a component of $G'$ isomorphic to $K_2$ and $w$ and $w'$ are internal vertices of $G$, each of them has a neighbor in $A\cup B$. Consequently, $A\cup B$ is an independent set in $G$ that dominates $\{z,z',w,w'\}$. If $I$ is a maximal independent set in $G$ containing $A\cup B$, then $|I\cap V(G')|\leq \alpha(G')-2$ since $I\cap V(G')$ contains neither of $z$, $z'$, $w$, $w'$ while any maximum independent set in $G'$ has to contain one of $z$, $z'$ and one of $w$, $w'$. Recall that $\alpha(G)=\alpha(G')+|A|+|B|+4$. Therefore, $|I|<\alpha(G)-1$, which contradicts the assumption that $G$ is almost well-covered.

$(\Leftarrow)$ Let $G$ be a connected graph of girth at least $6$ with exactly two vertices $x,y$ of type 2 which are adjacent, and with no type $k$ vertices for $k\geq 3$. Let $x',x''$ and $y',y''$ be the two leaves adjacent to $x$ and $y$, respectively. Suppose first that $V(G_0)=\emptyset$.
Let $I$ be a maximal independent set in $G$. If $I\cap \{x,y\}=\emptyset$,  then it is easy to see that $|I|=|V(G_1)|+4$, and if $I\cap \{x,y\} \neq \emptyset$,  then $|I|=|V(G_1)|+3$. Hence $\mu_\alpha(G)\leq 1$, and by Lemma~\ref{lem:leafcondition} the existence of type 2 vertices implies that $G$ is not well-covered. Hence $G$ is almost well-covered.

Suppose now that condition \textit{2.}~from the theorem holds, and let $u$ and $v$ be the two vertices of $G_0$.
Since neither of $x$ and $y$ has a neighbor in $G_0$ by Lemma \ref{lem:key2G_2} and $u$ and $v$ are of degree $2$ in $G$,
we may assume without loss of generality that the induced $6$-cycle containing $u$, $v$, $x$, and $y$ must be of the form
$x,a,u,v,b,y$ for some $a\in N(x)\cap V(G_1)$ and $b\in N(y)\cap V(G_1)$, see Fig.~\ref{fig:2type2-K2} for an illustration.
\begin{figure}[htpb]
\begin{center}
\includegraphics[width=0.5\linewidth]{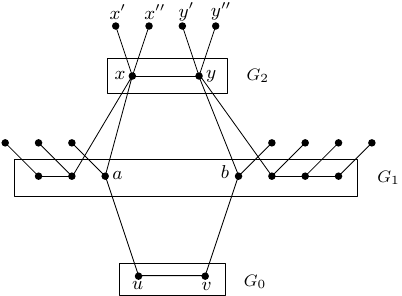}
\caption{Illustration of a graph with $G_0\cong K_2$.}\label{fig:2type2-K2}
\end{center}
\end{figure}
Let $I$ be a maximal independent set in $G$,  $I'=I\cap \{x,x',x'',y,y',y''\}$ and $I''=I\setminus I'$.
Suppose first that $I\cap \{u,v\}=\emptyset$. Then since $I$ is maximal, it follows that both $a$ and $b$ are contained in $I$.
Therefore, $I\cap \{x,y\}=\emptyset$, implying that $|I'|=4$. Since $I\cap \{u,v\}=\emptyset$, we get $|I''|=|V(G_1)|$ and hence $|I|= |V(G_1)|+4$.
Suppose now that $I\cap \{u,v\}\neq\emptyset$. In this case, $|I''|=|V(G_1)|+1$. Moreover, it is clear that $|I'|\in \{3,4\}$.
Thus, in both cases we have $|I|\in \{|V(G_1)|+4,|V(G_1)|+5\}$,
and the existence of type 2 vertices implies that $G$ is almost well-covered.
\end{proof}

Note that the two conditions from Theorem~\ref{thm:2G2Girth6} can be efficiently tested.
Therefore, using Corollary~\ref{cor:components} and Proposition~\ref{prop:wc_girth6},
we obtain the following algorithmic consequence.

\begin{corollary}\label{cor:polyrecognition2type2}
Given a graph $G$ of girth at least $6$ and with exactly two type 2 vertices, it can be decided in polynomial time whether $G$ is almost well-covered.
\end{corollary}

\subsection{Graphs of girth at least 6 with exactly one vertex of type 2}\label{subsec:1type2}

We now study almost well-covered graphs of girth at least 6 and with exactly one type 2 vertex. A characterization of such graphs is given by the following theorem.

\begin{theorem}\label{thm:1type2}
Let $G$ be a graph of girth at least 6, with exactly one type $2$ vertex $x$, and with no type $k$ vertices for $k\geq 3$.
Then $G$ is almost well-covered if and only if every internal vertex of $G$ other than $x$ is of type 0 or 1, $N(x)\cap V(G_0)=\emptyset$, $G_0$ is well-covered, and for every independent set $I$ in $G_1$ the following holds:
\begin{enumerate}[(a)]
\item $\alpha(G_0)=i(G_0-N(I))$ if $I\cap N(x)=\emptyset$ and
\item $\alpha(G_0)-i(G_0-N(I))\leq 1$ if $I\cap N(x)\neq \emptyset$.
\end{enumerate}
\end{theorem}

\begin{proof}
Let $G$ be a graph satisfying the hypotheses of the theorem.

$(\Rightarrow)$ Let $G$ be almost well-covered and let $x'$ and $x''$ denote the two leaves adjacent to $x$.
By Lemma~\ref{lem:key2G_2} it follows that $N(x)\cap V(G_0)=\emptyset$. It is also clear that every internal vertex different from $x$ is of type 1 or type 0.
Observe that $\alpha(G)=\alpha(G_0)+|V(G_1)|+2$.
We first claim that $G_0$ is well-covered. Suppose that there exist two maximal independent sets $I_1$ and $I_2$ of $G_0$ with $|I_1|>|I_2|$. By extending $I_1\cup \{x',x''\}$ and $I_2\cup \{x\}$ to maximal independent sets in $G$, we obtain two maximal independent sets in $G$ with sizes that differ by at least two, contradicting the assumption that $G$ is almost well-covered. Hence, $G_0$ is well-covered.

Let $I$ be an independent set in $G_1$. It is clear that $\alpha(G_0)\geq i(G_0-N(I))$. Suppose first that $I\cap N(x)=\emptyset$ and   $\alpha(G_0)\neq i(G_0-N(I))$. Then $ \alpha(G_0)> i(G_0-N(I))$. Let $S$ be a maximal independent set in $G_0-N(I)$ of size $i(G_0-N(I))$. Extending $S\cup I\cup \{x\}$ to a maximal independent set in $G$ results in a maximal independent set of $G$ of size less than $\alpha(G_0)+|V(G_1)|+1=\alpha(G)-1$, contradicting the assumption that $G$ is almost well-covered. Hence, condition $(a)$ holds.

Suppose now that $I\cap N(x)\neq \emptyset$. Let $S$ be a maximal independent set in $G_0-N(I)$ of minimum size, i.e., of size $i(G_0-N(I))$. Then every maximal independent set of $G$ containing $S\cup I$ is of size $|S|+|V(G_1)|+2$. Since $G$ is almost well-covered, it follows that
$$|S|+|V(G_1)|+2\geq\alpha(G)-1=\alpha(G_0)+|V(G_1)|+1,$$ hence $\alpha(G_0)-|S|=\alpha(G_0)-i(G_0-N(I))\leq 1$, that is, condition $(b)$ holds.

\smallskip
$(\Leftarrow)$
Suppose that every internal vertex of $G$ other than $x$ is of type 0 or 1, $N(x)\cap V(G_0)=\emptyset$, $G_0$ is well-covered, and for every independent set $I$ in $G_1$ conditions $(a)$ and $(b)$ hold.
Since $G$ has a vertex of type $2$, Lemma~\ref{lem:leafcondition} implies that $G$ cannot be well-covered.
Therefore, to show that $G$ is almost well-covered, it suffices to show that every maximal independent set in $G$ is of size at least $\alpha(G)-1$.
Let $S$ be a maximal independent set in $G$. Let $I_1=S\cap V(G_1)$ and $I_0=S\cap V(G_0)$.
Observe that
$$|S| = |I_0|+|V(G_1)|+\left\{
  \begin{array}{ll}
    1, & \hbox{if $x\in S$;} \\
    2, & \hbox{otherwise.}
  \end{array}
\right.$$
 It is clear that $I_0$ is a maximal independent set in $G_0-N(I_1)$, since $N(x)\cap V(G_0)=\emptyset$. Suppose first that $x\in S$. Then  $I_1\cap N(x)=\emptyset$. Since $G_0 - N(I_1)$ is an induced subgraph of $G_0$, we have $\alpha(G_0-N(I_1))\le \alpha(G_0)$. By the hypothesis $(a)$, it follows that $\alpha(G_0)=i(G_0 - N(I_1))$; therefore, $G_0 - N(I_1)$ is well-covered, in particular, $|I_0|=\alpha(G_0)$, which implies that $|S|=\alpha(G_0)+|V(G_1)|+1=\alpha(G)-1$.

 Suppose now that $x\not \in S$. If $N(x)\cap I_1=\emptyset$, then arguing similarly as above, we obtain that $|S|=\alpha(G)$. Finally, let $N(x)\cap I_1\neq \emptyset$. Then by hypothesis $(b)$ it follows that $|I_0|\geq \alpha
(G_0)-1$, and hence $|S|\geq \alpha(G)-1$. This shows that $G$ is almost well-covered.
\end{proof}

In the following lemma, we continue with the study of almost well-covered graphs of girth at least 6 with exactly one type 2 vertex in order to exhibit their structure in more detail and derive an efficient recognition algorithm for this case.

\begin{lemma}\label{lem:1type2}
If $G$ is an almost well-covered graph of girth at least $6$ with exactly one type 2 vertex, then every component of $G_0$ is isomorphic to $K_1$, $K_2$, $P_4$, or $C_7$.
\end{lemma}

\begin{proof}
Let $G$ be an almost well-covered graph of girth at least $6$ with exactly one type 2 vertex $x$. Let $x'$ and $x''$ denote the two leaves adjacent to $x$.
By Theorem~\ref{thm:1type2} it follows that $G_0$ is well-covered, and hence by Proposition~\ref{prop:wc_girth6}, it follows that every component of $G_0$ is isomorphic to $K_1$, $C_7$, or a graph with pendant edges forming a perfect matching.

Suppose now that $H$ is a component of $G_0$ that is isomorphic to neither of $K_1$, $K_2$, $P_4$, or $C_7$. Then there exists a perfect matching in $H$ formed by pendant edges. Suppose that $h_1h'_1, \ldots,h_mh'_m$ are pendant edges forming a perfect matching in $H$, where $h'_i$ are leaves in $H$. Observe that $m\geq 3$ and $\alpha(H)=m$. Since $H$ is connected, we may assume without loss of generality that $h_2$ is adjacent to $h_1$ and $h_3$. Since $h'_1,h'_2,h'_3\in V(G_0)$, it follows that there exist $a_1,a_2,a_3\in V(G_1)$ (where $a_1=a_3$ is possible) such that $h'_ia_i\in E(G)$. Suppose that $a_1\not \in N(x)$ and let $I=\{a_1\}$. Since the girth of $G$ is at least $6$, it follows that $a_1$ is not adjacent to $h_2$. It is now clear that every maximal independent set of $H-N(I)$ containing $h_2$ is of size at most $m-1<\alpha(H)$. Therefore, $i(H-N(I))<\alpha(H)$ and consequently $i(G_0-N(I))<\alpha(G_0)$, which violates condition $(a)$ from Theorem~\ref{thm:1type2}. We conclude that $a_1\in N(x)$. Similarly $a_3\in N(x)$. Using the fact that the girth of $G$ is at least 6, it follows that $a_1$ is not adjacent to $h_3$ and $a_3$ is not adjacent to $h_1$. Now let $I=\{a_1,a_3\}$ be an independent set in $G_1$. Observe that $h_2\in H-N(I)$. If $S$ is a maximal independent set in $H-N(I)$ containing $h_2$, then it is clear that $S\cap \{h_1,h'_1,h_3,h'_3\}=\emptyset$. This implies that $i(H-N(I))\leq \alpha(H)-2$, and consequently $i(G_0-N(I))\leq \alpha(G_0)-2$, which violates condition $(b)$ from Theorem~\ref{thm:1type2}.
\end{proof}

Lemma~\ref{lem:Hdomination}, Theorem~\ref{thm:1type2}, and Lemma~\ref{lem:1type2} imply the following.

\begin{theorem}\label{thm:1type2Poly}
Given a graph $G$ of girth at least $6$ with exactly one type 2 vertex, it can be decided in polynomial time whether $G$ is almost well-covered.
\end{theorem}

\begin{proof}
Let $G$ be a graph of girth at least $6$ with exactly one type 2 vertex.
If $G_0$ is the null graph then conditions $(a)$ and $(b)$ from Theorem~\ref{thm:1type2} are trivially satisfied, hence $G$ is almost well-covered.
If $N(x)\cap V(G_0)\neq \emptyset$ then by Theorem~\ref{thm:1type2} it follows that $G$ is not almost well-covered, hence we may assume that $N(x)\cap V(G_0)= \emptyset$.
If there is a component in $G_0$ not isomorphic to $K_1$, $K_2$, $P_4$, or $C_7$, then by Lemma~\ref{lem:1type2} it follows that $G$ is not almost well-covered. Suppose now that every component of $G_0$ is one of $K_1$, $K_2$, $P_4$, and $C_7$. Let $n_0$ be the maximum number of vertices in a component of $G_0$. Then $n_0\leq 7$.

For every independent set $I_1$ in $G_1$ of size at most $14$ we compute the value of $\alpha(G_0)-i(G_0-N(I_1))$. If for some such $I_1$ we have $\alpha(G_0)-i(G_0-N(I_1))\geq 2$, then by Theorem~\ref{thm:1type2} it follows that $G$ is not almost well-covered.
Next  we compute the value of $\alpha(G_0)-i(G_0-N(I_2))$ for every independent set $I_2$ in $G_1-N(x)$ with $|I_2|\leq 7$. If for some such $I_2$ we have $\alpha(G_0)-i(G_0-N(I_2))\geq 1$, then $G$ is not almost well-covered by Theorem~\ref{thm:1type2}.

Suppose now that for every independent set $I_1$ in $G_1$ of size at most 14 we have $\alpha(G_0)-i(G_0-N(I_1))\leq 1$ and for every independent set $I_2$ in $G_1- N(x)$ of size at most 7 we have $\alpha(G_0)=i(G_0-N(I_2))$. We claim that $G$ is almost well-covered. Suppose to the contrary that $G$ is not almost well-covered. Then it follows that at least one of the conditions $(a)$ or $(b)$ from Theorem~\ref{thm:1type2} is not satisfied. Suppose first that there exists an independent set $I$ in $G_1$ that doesn't satisfy condition $(b)$ of Theorem~\ref{thm:1type2}, that is $I\cap N(X)\neq \emptyset$ and $\alpha(G_0)-i(G_0-N(I))\geq 2$. Applying Lemma~\ref{lem:Hdomination} with $H=G_0$ and $k=2$, we can conclude that there exists an independent set $I_1\subseteq I$ such that $|I_1|\leq 14$ and $\alpha(G_0)-i(G_0-N(I_1))\geq 2$, a contradiction. Suppose now that there exists an independent set $I$ in $G_1$ which doesn't satisfy condition $(a)$ of Theorem~\ref{thm:1type2}, that is, $I\cap N(x)=\emptyset$ and $\alpha(G_0)-i(G_0-N(I))\geq 1$. Applying Lemma~\ref{lem:Hdomination} with $H=G_0$ and $k=1$, it follows that there exists an independent set $I_2\subseteq I$ of size at most 7 such that $\alpha(G_0)-i(G_0-N(I_2))\geq 1$. Observe that since $I_2\subseteq I$ and $I\cap N(x)=\emptyset$ it follows that $I_2\cap N(x)=\emptyset$. Hence $I_2$ is an independent set in $G_1-N(x)$ of size at most $7$ such that $\alpha(G_0)-i(G_0-N(I_2))\geq 1$, contradicting the assumption at the beginning of the paragraph. This proves our claim that $G$ is well-covered.

Since we only need to check independent sets of $G_1$ of bounded size, the algorithm can be implemented to run in polynomial time.
\end{proof}

\subsection{Almost well-covered \textbf{$\{C_3,C_4,C_5,C_7\}$}-free graphs}\label{subsec:notype2}

Based on results in the previous subsections, we now develop a polynomial-time recognition algorithm for almost well-covered
$\{C_3,C_4,C_5,C_7\}$-free graphs. The following lemma will be crucial.

\begin{lemma}\label{lem:deg2}
Let $G$ be an almost well-covered $\{C_3,C_4,C_5,C_7\}$-free graph. Then for all $x \in V(G_0)$, we have $d_{G_0}(x)\leq 2$.
\end{lemma}

\begin{proof}
Assume for a contradiction that there is a vertex $x\in V(G_0)$ such that $d_{G_0}(x)= k\geq 3$. Let $x_1, \ldots ,x_k$ be the neighbors of $x$ in $G_0$.
Let $I$ be the set of vertices of $G$ at distance $3$ from $x$. Since there are no cycles of length $3$, $5$, or $7$ in $G$, it follows that $I$ is an independent set. Consider now the graph $G'=G-N[I]$.
Since each of $x_i$ is of type 0 and $G$ is $\{C_3,C_4,C_5\}$-free, it follows that $N[I]$ contains all neighbors of the $x_i$s except $x$,  and consequently each of $x_i, i=1, \ldots ,k$ is a leaf in $G'$ (note also that $x_i$s are not adjacent to each other since otherwise we would have a triangle). Then $x$ has at least $k\geq 3$ leaves in $G'$, hence by Corollary \ref{cor:atmost2leaves}, $G'$ has independence gap at least 2, which contradicts Lemma~\ref{tool}.
\end{proof}

\begin{corollary}\label{cor:type1G_0}
Let $G$ be an almost well-covered $\{C_3,C_4,C_5,C_7\}$-free graph and with no vertex of type 2. Then every component of $G_0$ is isomorphic to one of $P_1$, $P_2$, $P_3$, $P_4$, $P_5$,  $P_6$, $P_7$, $P_8$, $P_{10}$, $C_6$, $C_8$, $C_9$, $C_{10}$, $C_{11}$, $C_{13}$.
\end{corollary}

\begin{proof}
Lemma~\ref{lem:deg2} implies that every component of $G_0$ is a cycle or a path. By Corollary~\ref{cor:type1main} it follows that $\mu_\alpha(G_0)\leq 1$. Hence, every component of $G_0$ is a cycle or a path which is well-covered or almost well-covered. It can easily be seen that there is no allowed well-covered cycle, and the only well-covered paths are $P_1,P_2$, and $P_4$. The only almost well-covered cycles are $C_6$, $C_8$, $C_9$, $C_{10}$, $C_{11}$, $C_{13}$ and the only almost well-covered paths are $P_3$, $P_5$,  $P_6$, $P_7$, $P_8$, $P_{10}$. The claimed result follows.
\end{proof}

\begin{theorem}\label{thm:polyrecognition}
Given a $\{C_3,C_4,C_5,C_7\}$-free graph, it can be decided in polynomial time whether $G$ is almost well-covered.
\end{theorem}

\begin{proof}
Let $G$ be a $\{C_3,C_4,C_5,C_7\}$-free graph.
If $G$ is disconnected then by Corollary~\ref{cor:components} $G$ is almost well-covered if and only if all components of $G$ are well-covered, except from one which is almost well-covered. By Proposition~\ref{prop:wc_girth6}, it follows that a component of $G$ is well-covered if and only if it is isomorphic to $K_1$ or its pendant edges form a perfect matching.
It is now clear that verifying if a given component of $G$ is well-covered can be done in polynomial time.
Therefore, we may assume that $G$ is connected and not well-covered.
If there is an internal vertex of $G$ that is not of type 0, 1, or 2, then $G$ is not almost well-covered by Corollary~\ref{cor:atmost2leaves}. If $V(G_2)$ is not a clique, then by Lemma~\ref{lem:type 2 clique}, $G$ is not almost well-covered.
From now on assume that $V(G_2)$ is a clique. Since $G$ is $C_3$-free, there are at most two vertices of type 2. If there are exactly two type 2 vertices, the result follows by Corollary~\ref{cor:polyrecognition2type2}. If there is exactly one type 2 vertex, the result follows by Theorem~\ref{thm:1type2Poly}. Suppose now that no vertex of $G$ is of type $2$.
Since $G$ is not well-covered, we have that $V(G_0)\neq\emptyset$ by Proposition~\ref{prop:wc_girth6}.
If $G_0$ has a component having at least $14$ vertices, then, using Corollary~\ref{cor:type1G_0}, we infer that $G$ is not almost well-covered.
Suppose now that the maximum number of vertices in a component of $G_0$ is at most 13.
Since $G$ is not well-covered, the condition that $G$ is almost well-covered is equivalent to the condition $\mu_{\alpha}(G)\leq 1$.
By Corollary~\ref{cor:G_0boundedPoly}, since the order of every component of $G_0$ is bounded by a constant, this latter condition can be checked in polynomial time. This completes the proof.
\end{proof}

Our work leaves open characterizations and/or polynomial-time recognition algorithms for almost well-covered graphs of girth at least $k$ for $k\in \{4,5,6,7\}$. In particular, since well-covered graphs of girth at least $6$ are efficiently characterized (cf.~Proposition~\ref{prop:wc_girth6}) and by Corollary~\ref{cor:polyrecognition2type2} and Theorem~\ref{thm:1type2Poly} there exists a polynomial-time recognition algorithm for almost well-covered graphs of girth at least $6$ with at least one type $2$ vertex, we find the following problem particularly interesting.

\begin{openProblem}
Can almost well-covered graphs of girth at least $6$ be recognized in polynomial time?
\end{openProblem}

\section*{Acknowledgments}
We would like to thank the anonymous referees for careful reading of the manuscript and several helpful comments.

\bibliographystyle{abbrvnat}
\bibliography{biblio}
\label{sec:biblio}

\end{document}